\documentclass[document, 11]{article}
\usepackage[utf8]{inputenc}

\usepackage[scaled]{helvet} 
\usepackage{courier} 
\normalfont
\usepackage[T1]{fontenc}

\title{Approximating Unique Games Using Low Diameter\\ Graph Decomposition}
\author{Vedat Levi Alev\footnote{Supported by the GO-Bell Scholarship and the David R.~Cheriton Graduate Scholarship. E-mail: \texttt{vlalev@uwaterloo.ca}}\\
University of Waterloo
\and Lap Chi Lau\footnote{Supported by NSERC Discovery Grant 2950-120715 and NSERC Accelerator Supplement 2950-120719. E-mail: \texttt{lapchi@uwaterloo.ca}}\\
University of Waterloo}
\date{}

\usepackage{boxedminipage}
\usepackage{amsmath, amssymb, mathtools,etoolbox, mathrsfs, float, amsthm, graphicx, subcaption}

\usepackage[usenames,dvipsnames,svgnames,table]{xcolor}

\usepackage[colorlinks=true,
            linkcolor=NavyBlue,
            urlcolor=blue,
            citecolor=Blue]{hyperref}

\def\UG{\mathsf{UG}}
\def\2Lin{\mathsf{2Lin}}
\def\M2Lin{\mathsf{Max}\textrm{-}\2Lin}
\def\Mi2Lin{\mathsf{Min}\textrm{-}\2Lin}
\def\G2Lin{\Gamma\textrm{-}\2Lin}
\DeclareMathOperator{\SAT}{\mathsf{SAT}}
\DeclareMathOperator{\UNSAT}{\mathsf{UNSAT}}

\def\Ins{\mathfrak I}
\theoremstyle{plain}
\newtheorem{theorem}{Theorem}[section]
\newtheorem{conjecture}[theorem]{Conjecture}
\newtheorem{lemma}[theorem]{Lemma}

\newtheorem{corollary}[theorem]{Corollary}
\theoremstyle{definition}
\newtheorem{definition}[theorem]{Definition}

\theoremstyle{enumtheo}

\newtheoremstyle{break}
  {\topsep}
  {\topsep}
  {}
  {}
  {\bfseries}
  {.}
  {\newline}
  {}
\theoremstyle{break}
\newtheorem{algo}[theorem]{Algorithm}

\newenvironment{algorithm}[3]
        {\begin{boxedminipage}{\textwidth}\begin{algo}[#1]
        {\begin{tabular}{r l}
        \textbf{Intput} & #2\\
        \textbf{Output} & #3
        \end{tabular}\par\enskip}}
        {\end{algo}\end{boxedminipage}}

\usepackage{tikz}
\usetikzlibrary{shapes, fit}

\def\bNP{\mathbf{NP}}
\def\RR{\mathbb{R}}
\def\ZZ{\mathbb{Z}}

\def\one{\mathbf 1}
\def\LP{\mathsf{LP}}
\def\ee{\varepsilon}
\def\Ins{\mathfrak{I}}
\def\Par{\mathcal{A}}
\def\Cc{\mathcal{C}}

\DeclareMathOperator{\Exp}{\mathbb E}
\DeclareMathOperator{\Pp}{\mathbb P}
\DeclareMathOperator{\OO}{\mathcal O}
\DeclareMathOperator{\poly}{poly}
\DeclarePairedDelimiter\set{\lbrace}{\rbrace}
\DeclarePairedDelimiter\sqbr{[}{]}
\DeclarePairedDelimiter\Abs{|}{|}

\begin{document}

\maketitle

\begin{abstract}
We design approximation algorithms for Unique Games when the constraint graph admits good low diameter graph decomposition.
For the $\M2Lin_k$ problem in $K_r$-minor free graphs, when there is an assignment satisfying $1-\ee$ fraction of constraints, we present an algorithm that produces an assignment satisfying $1-O(r\ee)$ fraction of constraints, with the approximation ratio independent of the alphabet size.
A corollary is an improved approximation algorithm for the {\sf Min-UnCut} problem for $K_r$-minor free graphs.
For general Unique Games in $K_r$-minor free graphs, we provide another algorithm that produces an assignment satisfying $1-O(r \sqrt{\ee})$ fraction of constraints.

Our approach is to round a linear programming relaxation to find a
minimum subset of edges that intersects all the inconsistent
cycles. We show that it is possible to apply the low diameter graph
decomposition technique on the constraint graph directly, rather than
to work on the label extended graph as in previous algorithms for
Unique Games.  The same approach applies when the constraint graph is
of genus $g$, and we get similar results with $r$ replaced by $\log g$
in the $\M2Lin_k$ problem and by $\sqrt{\log g}$ in the general
problem.  The former result generalizes the result of Gupta-Talwar for
Unique Games in the $\M2Lin_k$ case, and the latter result generalizes
the result of Trevisan for general Unique Games.
\end{abstract}

\section{Introduction}
For a given integer $k \ge 1$, an undirected graph $G = (V, E)$ and a
set $\Pi = \set{\pi_{uv} : uv \in E}$ of permutations on $[k]$
satisfying $\pi_{uv} = \pi_{vu}^{-1}$, the Unique Games problem with
alphabet size $k$ (denoted by $\UG_k$) is the problem of finding an
assignment $x: V \to [k]$ to the vertices such that the number of
edges $e = uv \in E$ satisfying the constraint $\pi_{uv}(x(u)) = x(v)$
is maximized.  The value $\SAT(\Ins)$ of a Unique Games instance $\Ins
= (G,\Pi)$ is defined as,
\[ \SAT(\Ins) = \max_{x:V \to [k]}
  \frac{1}{|E|} \sum_{uv \in E}\one[\pi_{uv}(x(u)) =
    x(v)]\]
i.e.~ the maximum fraction of satisfiable constraints over all
assignments $x$.  We define $\UNSAT(\Ins) = 1 - \SAT(\Ins)$ as the
minimum fraction of unsatisfied constraints.

The Unique Games Conjecture of Khot \cite{KhotUG} postulates that
it is $\bNP$-hard to distinguish whether a given instance $\Ins = (G,
\Pi)$ of the Unique Games problem is almost satisfiable or
almost unsatisfiable,
and the problem becomes harder as the alphabet size $k$ increases.
\begin{conjecture}[The Unique Games Conjecture, \cite{KhotUG}]\label{conj:ugc}
  For every $\ee > 0$, there exists an integer $k := k(\ee)$, such
  that the decision problem of whether an instance $\Ins$ of $\UG_{k}$
  satisfies $\SAT(\Ins) \geq 1-\ee$ or $\SAT(\Ins) \leq \ee$ is
  $\bNP$-hard.
\end{conjecture}
The Unique Games Conjecture has attracted much attention over the
years, due to its implications regarding the hardness of approximation
for many $\bNP$-hard problems~\cite{KhReVC,G2LinHard,RaghUGC}.  An
important case of Unique Games is the $\M2Lin_k$ problem when the
constraints are of the form $x_u - x_v \equiv c_{uv} \pmod k$ for $uv
\in E$.  This problem is shown to be as hard as the general case of
the Unique Games problem by Khot et al.~\cite{G2LinHard}.  The {\sf
  Max-Cut} problem is a well-studied special case of $\M2Lin_2$ where
$x_u - x_v \equiv 1 \pmod 2$ for $uv \in E$.  Assuming the Unique
Games Conjecture, Khot et al.~\cite{G2LinHard} proved that it is
$\bNP$-hard to distinguish {\sf Max-Cut} instances where the optimal
value is at least $1-\ee$ from instances where the optimal value is at
most $1-\Theta(\sqrt{\ee})$.


There have been several efforts in designing polynomial time
approximation algorithms for Unique
Games~\cite{KhotUG,TrevUGC,GTUGC,CMM0,CMM}, where the objective is to
minimize the number of unsatisfied constraints.  Let $\Ins$ be the
given instance of $\UG_k$ with $n$ variables and $\UNSAT(\Ins)=\ee$.
Trevisan~\cite{TrevUGC} gave an SDP-based algorithm that provides an assignment which violates at most an $\OO(\sqrt{\ee \log n})$ fraction of the constraints.
Gupta and Talwar~\cite{GTUGC} gave an LP-based algorithm that provides
an assignment which violates at most an $\OO(\ee \cdot \log n)$
fraction of the constraints.  Charikar, Makarychev, and Makarychev
\cite{CMM0} gave an SDP-based algorithm which finds an assignment
violating at most a $\OO(\sqrt{\ee \log k})$ fraction of constraints,
where $k$ is the alphabet size.  Chlamtac, Makarychev, and Makarychev
\cite{CMM} gave another SDP-based algorithm which finds an assignment
violating at most an $\OO(\ee \cdot \sqrt{\log k \log n})$-fraction of
the constraints.

There are also some previous works exploiting the structures of
the constraint graphs.  Arora, Barak and Steurer~\cite{ABS} presented
a sub\-exponential time algorithm to distinguish the two cases in the
Unique Games Conjecture.  Their approach uses the spectral information
of the constraint graph.  If the Laplacian matrix of the constraint
graph has only a few small eigen\-values, then they extend the
subspace enumeration approach of Kolla~\cite{Kolla} to search over
this eigen\-space for a good assignment.  On the other hand, if there
are many small eigen\-values, they give a graph decomposition procedure
to delete a small fraction of edges so that each component in the
remaining graph has only a few small eigen\-values.  Combining these two
steps carefully gives their sub\-exponential time algorithm. There is
also an SDP-based propagation rounding approach to find a good
assignment when the constraint graph is an expander~\cite{ExpUGC} and
more generally when the Laplacian matrix of the constraint graph has
only a few small eigen\-values~\cite{BRS,GS}. These gave an alternative
SDP-based sub\-exponential time algorithm for the Unique Games
Conjecture.


Our initial motivation is to study the Unique Games problem when the
Laplacian matrix of the constraint graph has many small eigen\-values,
as there are no known good approx\-imation algorithms for Unique Games
in these graphs.  The most natural graph family possessing this
property is the class of graphs without a $K_r$ minor, where a graph
$H$ is a minor of $G$ if $H$ can be obtained from $G$ by deleting and
contracting edges, and $K_r$ is the complete graph with $r$ vertices.
Kelner et al.~\cite{Deform2}, after a sequence of
works~\cite{Deform1,SpectralWorks,DeformBG}, proved that the $k$-th
smallest eigen\-value of the Laplacian matrix of a bounded degree
$K_r$-minor free graph is $\OO(\poly(r)\cdot k/n)$, showing that there are
many small eigen\-values.  The class of $K_r$-minor free graphs is well
studied and is known to contain the class of planar graphs and the
class of bounded genus graphs, where a graph is of genus $g$ if the
graph can be embedded into a surface having at most $g$ handles
without edge crossings.  There are different (non-spectral) techniques
in designing approximation algorithms for various problems in
$K_r$-minor free graphs (see e.g.~\cite{BiDim,DHK}), including
problems that are known to be harder than Unique Games.  This leads us
to the question of whether we can extract those ideas to design better
algorithms for Unique Games.


\subsection{Our Results}

In this paper, we consider the problem of approximately minimizing the
number of unsatisfied constraints in an $\UG_k$ instance $\Ins = (G, \Pi)$,
when the constraint graph $G$ is $K_r$-minor free.
Our first theorem is for the $\M2Lin_k$ problem.

\begin{theorem}\label{thm:g2linalg}
Given a $\M2Lin_k$ instance $\Ins = (G, \Pi)$ where $G$ is a
$K_r$-minor free graph and $\UNSAT(\Ins) = \ee$ (respectively where
$G$ is of genus at most $g$), there is an LP-based polynomial time
algorithm which outputs an assignment that violates at most an $\OO(r
\cdot \ee)$ fraction of constraints (respectively at most a $\OO(\log
g \cdot \ee)$ fraction of constraints).
\end{theorem}

For {\sf Max-2Lin}, Theorem~\ref{thm:g2linalg} on bounded genus graphs is a refinement of the $\OO(\log n \cdot \ee)$ bound of Gupta and Talwar~\cite{GTUGC} as $g = \OO(n)$.
Theorem~\ref{thm:g2linalg} also implies an improved approximation algorithm for the {\sf Min-Uncut} problem (the complement of the {\sf Max-Cut} problem), where the objective is to delete a minimum subset of edges so that the resulting graph is bipartite.

\begin{corollary} \label{cor:minuncut}
There is an LP-based polynomial time $O(r)$-approximation algorithm (respectively a $O(\log g)$-approximation algorithm) for the {\sf Min-Uncut} problem for $K_r$-minor free graphs (respectively for graphs of genus $g$).
\end{corollary}

The best known approximation algorithm for {\sf Min-Uncut} is an
SDP-based $\OO(\sqrt{\log n})$-approximation algorithm~\cite{UncutSDP,CMM}.
We are not aware of any improvement of this bound for $K_r$-minor free
graphs and bounded genus graphs.  The above algorithms crucially used
the symmetry of the linear constraints in {\sf Max-2Lin}.  For general
Unique Games, we present a different algorithm with weaker guarantees.
The following theorem on bounded genus graphs is a refinement of the $\OO(\sqrt{\ee \cdot \log n})$ bound of Trevisan~\cite{TrevUGC} (see the discussion in~\cite[Section 4]{GTUGC}).

\begin{theorem}\label{thm:ugkalg}
Given a $\UG_k$ instance $\Ins = (G, \Pi)$ where $G$ is a $K_r$-minor free graph and $\UNSAT(\Ins) = \ee$ (respectively where $G$ is of genus at most $g$),
there is an LP-based polynomial time algorithm which outputs an assignment that violates at most an $\OO(r \cdot \sqrt{\ee})$ fraction of constraints (respectively at most a $\OO(\sqrt{\log g \cdot \ee})$ fraction of constraints).
\end{theorem}

The main tool in our algorithms is the low diameter graph
decomposition for $K_r$-minor free graphs and bounded genus graphs
(see Section~\ref{sec:decomp}).  Both of our algorithms are LP-based.
The $\M2Lin_k$ algorithm is based on cutting inconsistent cycles,
which is different from most existing algorithms for Unique Games that
are based on finding good assignments.  The $\UG_k$ algorithm is based
on the propagation rounding method in Gupta and Talwar~\cite{GTUGC}.
We defer the technical overviews to Section~\ref{sec:overview1} and
Section~\ref{sec:overview2}, after the preliminaries are defined.

\subsection{Related Work}

There are polynomial time approximation schemes for many problems in
$K_r$-minor free graphs (see~\cite{DHK,BiDim}). For example, there is
a $(1-\ee)$-approximation algorithm for {\sf Max-Cut} with running
time $2^{1/\ee} \cdot n^{\OO_r(1)}$ for $K_r$-minor free graphs. The
approach is a generalization of Baker's approach for planar graphs~\cite{Baker}
: removing a small fraction of edges so that the remaining graph is of
bounded treewidth, and then using dynamic programming to solve the
problem on each bounded treewidth component.  This approach can be
used to distinguish the two cases in the Unique Games Conjecture for
$K_r$-minor free graphs for any fixed $r$.  However, this approach is
not applicable to obtain multiplicative approximation algorithms for
minimizing the number of unsatisfied constraints for Unique Games,
since it requires to remove a constant fraction of edges while the
optimal value could be very small.  As mentioned previously, we are
not aware of any polynomial time approximation algorithms with
performance ratio better than $\OO(\sqrt{\log n})$ for the {\sf
  Min-Uncut} problem for $K_r$-minor free graphs.

The low diameter graph decomposition technique is very useful in
designing approximation algorithms for $K_r$-minor free graphs.  It
was first developed by Bartal in \cite{Bartal} for probabilistically
approximating general metric spaces by tree metrics. The special case
of planar and minor-free graphs was studied by Klein, Plotkin and
Rao~\cite{KPRDecomp} to establish the multicommodity flow-cut gap of
$K_r$-minor free graphs, and since then this technique has found
numerous applications.  A recent result using this technique is a
$(\OO_{\ee}(r^2),1+\ee)$ bicriteria approximation algorithm~\cite{BFK}
for the small set expansion problem, which is shown to be closely
related to the Unique Games problem~\cite{RSSSE,RSTSSE}.

It is a well-known result of Hadlock~\cite{Hadlock} that the maximum
cut problem can be solved exactly in polynomial time on planar graphs.
In Agarwal's thesis~\cite{AgarwalUGC}, he showed that an SDP
relaxation (with triangle inequalities) for $\UG_2$ is exact for
planar graphs, using a multicommodity flow-cut type argument
introduced in Agarwal et al.~\cite{hypercube}.  It is mentioned
in~\cite{AgarwalUGC} that this approach of bounding the integrality gap
(even approximately) is only known to work for $K_5$-minor free graphs.

Steurer and Vishnoi~\cite{SteuVish} showed that the Unique
Games problem can be reduced to the {\sf Multicut} problem and used it
to recover Gupta and Talwar's result in the case of $\M2Lin_k$.  The
approach of Steurer and Vishnoi is similar to ours; see
Section~\ref{sec:overview1} for some discussion.

\subsection{Organization}

In Section~\ref{sec:decomp}, we describe the low diameter graph
decomposition results that we will apply.  In
Section~\ref{sec:minuncut}, we first present the proof for the {\sf
  Min-Uncut} problem, as it is simpler and illustrates all the main
ideas.  Then we generalize the proof to the $\M2Lin_k$ problem
in Section~\ref{sec:max2lin}.  In Section~\ref{sec:genug}, we show the
result for general Unique Games.  The proof overviews for
Theorem~\ref{thm:g2linalg} and Theorem~\ref{thm:ugkalg} will be
presented in the corresponding sections, Section~\ref{sec:overview1}
and Section~\ref{sec:overview2}, after the preliminaries are defined.

\section{Low Diameter Graph Decompositions}\label{sec:decomp}

Let $G = (V, E)$ be a graph with non-negative edge weights $w : E \to \RR_+$.
A collection $P = \set{C_1, \ldots, C_k}$ of disjoint subsets $C_j \subseteq V$ (called clusters) is a partition if they satisfy $V = \cup_{i = 1}^k C_j$.
We call a partition $P$ {weakly $\Delta$-bounded}
if each of the clusters has weak diameter $\Delta$, i.e.
$$d_G(u, v) \le \Delta \quad\forall u, v \in C_j; \forall j \in [k]$$
where $d_G$ denotes the shortest path distance on $G$ (induced by the
edge weights $w$).  We say that the partition $P$ is
{strongly $\Delta$-bounded} if each cluster has strong diameter
$\Delta$, i.e.
$$d_{G[C_j]}(u, v) \le \Delta \quad \forall u, v \in C_j; \forall j \in [k]$$
where $d_{G[C_j]}$ denotes the shortest path distance in the induced subgraph $G[C_j]$.
%
%
We write $P(u)$ for the unique cluster $C_j$ containing the vertex $u
\in V$. We call a distribution $\Par$ of partitions {$\Delta$-bounded
$D$-separating} if each cluster is of diameter $\Delta$ and for each
edge $uv \in E$ we have
\begin{equation} \label{e:cut}
  \Pp_{P \sim \Par}[P(u) \ne P(v)] \le \frac{D}{\Delta} \cdot w(u,v).
\end{equation}%
This implies that we can cut a graph into clusters with diameter at most $\Delta$ by deleting all the inter-cluster edges, while only losing a $D/\Delta$ fraction of the total edge weight.
%
%
We call a $\Delta$-bounded $D$-separating partitioning scheme efficient,
if we can sample it in polynomial time.

The seminal work of Klein, Plotkin and Rao~\cite{KPRDecomp} showed the
first low diameter graph decomposition scheme for planar graphs and
more generally for $K_r$-minor free graphs.  We will use the latest
result of this line of work~\cite{KPRDecomp,FaTa,LeeSid,CopsRobbers},
as it gives the best known quantitative bound and also it guarantees
the clusters have strong diameter $\Delta$ which will be important in
our algorithm for general Unique Games.

\begin{theorem}[\cite{CopsRobbers}]\label{thm:copsrobbers}
  Every weighted $K_r$-minor free graph admits an efficient weakly
  $\Delta$-bounded $\OO(r)$-separating partitioning scheme for any
  $\Delta \geq 0$.
\end{theorem}

\begin{theorem}[\cite{CopsRobbers}]\label{thm:crstrong}
  Every weighted $K_r$-minor free graph admits an efficient strongly $\Delta$-bounded $\OO(r^2)$-separating partitioning scheme for any $\Delta \geq 0$.
\end{theorem}

We will also use the optimal bounds for bounded genus graphs, to derive better results for Unique Games in these graphs.
\begin{theorem}[\cite{CopsRobbers,LeeSid}]\label{thm:cr2}
Every weighted graph of genus $g$ admits an efficient strongly $\Delta$-bounded $\OO(\log g)$-separating partitioning scheme for any $\Delta \geq 0$.
\end{theorem}
The results in~\cite{CopsRobbers} are stated using the language of
padded decompositions, but it is easy to see that the results we stated are corollaries of the theorems in~\cite{CopsRobbers}.

\section{Minimum Uncut} \label{sec:minuncut}

Given an undirected graph $G=(V,E)$ with a non-negative cost $c_e$ for
each edge $e \in E$, the {\sf Min-Uncut} problem is to find a subset
$S \subseteq V$ to minimize the total cost of the uncut edges (the
edges with both endpoints in $S$ or both endpoints in $V-S$).
Alternatively, the problem is equivalent to finding a subset
$F \subseteq E$ of minimum total cost so that $G-F$ is a bipartite
graph (so $F$ is the uncut edges).  As a graph is bipartite if and
only if it has no odd cycles, the problem is equivalent to finding a
subset of edges of minimum total cost that intersects all the odd
cycles in the graph, which is also known as the {\sf Odd Cycle
  Transversal} problem.  We will tackle the {\sf Min-Uncut} problem
using this perspective,
  by writing a linear program for the {\sf Odd Cycle
  Transversal} problem.

As mentioned already, the {\sf Min-Uncut} problem is a special case of
$\M2Lin_2$. We will see in Section~\ref{sec:max2lin} that the ideas in
this section can be readily generalized to design an approximation
algorithm for the $\M2Lin_k$ problem.

\subsection{Linear Programming Relaxation}\label{sec:LP}

We consider the following well-known linear programming relaxation for
the {\sf Odd Cycle Transversal} problem, which
is known to be exact when the
input is a planar graph~\cite{MinUncutPlanarLP}.
We note that this is similar to the LP formulation used by Gupta and Talwar~\cite{GTUGC}
when specialized to the {\sf Min-Uncut} problem,
but their LP formulation is on the ``label extended graph'' that we will explain soon.

\noindent
\begin{boxedminipage}{\textwidth}
  {\begin{flalign*}\label{lp:min-uncut}
    \LP^\star = \min \sum_{e \in E} c_e x_e \tag{LP-MinUncut} &&
  \end{flalign*}%
  subject to%
  \begin{flalign*}
    \sum_{e \in C} x_e \ge 1 & \quad C \in \Cc & \\
    x_e \ge  0 & \quad e \in E &
  \end{flalign*}%
  where $\Cc$ is the set of odd cycles of $G$.}
\end{boxedminipage}\newline
%

This LP has exponentially many
constraints.  To solve it in polynomial time using the ellipsoid
method~\cite{ellipsoidref}, we require a polynomial time separation
oracle to check whether a solution $x$ is feasible or not, and if not
provide a violating constraint.  For this LP, it is well known that
the separation oracle can be implemented in polynomial time using
shortest path computations (e.g. see~\cite{GTUGC}).  Since this will
be relevant to our discussion, we describe the separation oracle in
the following.

The idea is to construct the ``label extended graph'' $H=(V',E')$ of
$G=(V,E)$ (to use the Unique Games terminology).
For each vertex $v$ in $V$, we create two vertices $v^+$ and $v^-$ in
$V'$.  For each edge $uv$ in $E$, we add two edges $u^+ v^-$ and $u^-
v^+$ in $E'$, and we set the weight of $u^+ v^-$ and $u^- v^+$ to be
$x_{uv}$. By construction, there is an odd cycle in $G$ containing $v$
if and only if there is a path from $v^+$ to $v^-$ in the label extended
graph $H$.  So, to check that $x$ is feasible, we just need to check
that the weight of the shortest path from $v^+$ to $v^-$ is at least one
for every $v$.

\subsection{Proof Overview} \label{sec:overview1}

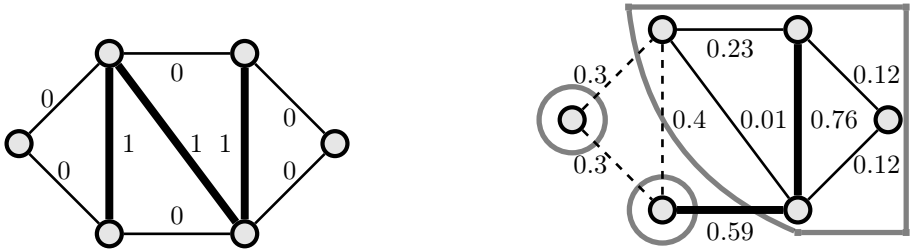
\begin{figure} 
  \centering
  \begin{subfigure}[t]{0.4\textwidth}
    \tikzstyle{normal edge} = [draw, line width=1pt, -, black]
    \tikzstyle{heavy edge} = [draw, line width=3pt, -, black]
    \begin{tikzpicture}[scale=0.6]
      \begin{scope}[every node/.style={draw, circle, fill=gray!20, ultra thick}]
        \node (a) at (0, 0) {};
        \node (b) at (2,2) {};
        \node (c) at (5,2) {};
        \node (d) at (2, -2) {};
        \node (e) at (5, -2) {};
        \node (f) at (7, 0) {};
      \end{scope}

      \begin{scope}
        \path[normal edge] (a) edge node[left] {0} (b)
        edge node[above] {0} (d);
        \path[normal edge] (b) edge node[below] {0} (c);
        \path[heavy edge]  (b) edge node[right] {1} (d);
        \path[normal edge] (d) edge node[above] {0} (e);
        \path[normal edge] (c) edge node[below] {0} (f);
        \path[heavy edge] (e) edge node[left] {1} (c);
        \path [heavy edge] (b) edge node[right] {1} (e);
        \path[normal edge] (e) edge node[above] {0}  (f);
      \end{scope}
    \end{tikzpicture}
    \caption{The shortest path distance between any two pairs of vertices is 0. The bold edges correspond to an optimal integral solution to \ref{lp:min-uncut}. }\label{fig:useless}
  \end{subfigure}\hspace{1em}
  \begin{subfigure}[t]{0.4\textwidth}
    \tikzstyle{normal edge} = [draw, line width=1pt, -, black]
    \tikzstyle{heavy} = [draw, line width=3pt, -, black]
    \tikzstyle{ic edge} = [draw, line width=1pt, -, black, dashed]
    \tikzstyle{invisible} = [fill=black!50, rectangle, inner sep=0.0000001pt, draw=black!50]
    \tikzstyle{separ} = [opacity=0.5, line width=2pt]
    \begin{tikzpicture}[scale=0.6]
      \begin{scope}[every node/.style={draw, circle, fill=gray!20, ultra thick}]
        \node (a) at (0, 0) {};
        \node (b) at (2,2) {};
        \node (c) at (5,2) {};
        \node (d) at (2, -2) {};
        \node (e) at (5, -2) {};
        \node (f) at (7, 0) {};
        \node[circle, fit=(a), fill opacity=0, draw=black!50, line width=2pt](aCluster) {};
        \node[circle, fit=(d), fill opacity=0,  draw=black!50, line width=2pt](aCluster) {};
        \node[invisible] (topLeft) at (1.25, 2.5) {};
        \node[invisible] (botLeft) at (5, -2.5) {};
        \node[invisible] (topRight) at (7.4, 2.5) {};
        \node[invisible] (botRight) at (7.4, -2.5) {};
      \end{scope}

    \begin{scope}
      \path[ic edge] (a) edge node[left, black](ab) {0.3} (b);
      \path[ic edge] (a) edge node[left, black, ](ad) {0.3} (d);
      \path[normal edge] (b) edge node[below](bc) {0.23} (c);
      \path[ic edge]  (b) edge node[right, black](bd) {0.4} (d);
      \path [normal edge] (b) edge node[right](be) {0.01} (e);
      \path[heavy] (d) edge node[below](de) {0.59} (e);
      \path[normal edge] (c) edge node[right](cf) {0.12} (f);
      \path[heavy] (e) edge node[right](ec) {0.76} (c);
      \path[normal edge] (e) edge node[right](ef) {0.12}  (f);
      \path[separ] (topLeft) edge (topRight) edge[bend right] (botLeft);
      \path[separ] (botLeft) edge (botRight);
      \path[separ] (botRight) edge (topRight);
    \end{scope}
  \end{tikzpicture}
  \caption{After removing edges with weight at least $1/2$ (the bold edges), all remaining subgraphs are of diameter at most $1/4$ and they are bipartite. The dashed edges are the inter-cluster edges.} \label{fig:useful}
  \end{subfigure}
\caption{Applying low diameter graph decomposition in a feasible solution to LP-MinUncut.}
\end{figure}

One natural approach to do the rounding is to consider the label extended graph $H$ of $G$.
From the above discussion,
destroying all the odd cycles in $G$ is equivalent to destroying all
the $v^+$-$v^-$ paths in $H$ for all $v$.
Since $x$ is feasible, we know that the shortest path distance between
$v^+$ and $v^-$ is at least $1$ for every $v$.  Therefore, we can
apply the low diameter graph decomposition result in the label
extended graph, by setting $\Delta<1$ to ensure that all $v^+$ and
$v^-$ are disconnected, and hope to delete edges with weight at most
$\sum_{e \in E} O(r/\Delta) \cdot c_e x_e = O(r) \cdot \LP^\star$ by
Theorem~\ref{thm:copsrobbers}.  This is similar to the approach used
in~\cite{SteuVish} to reduce Unique Games to {\sf Multicut}.  The
problem of this approach is that the label extended graph $H$ could
have arbitrarily large clique minor, even though the original
constraint graph $G$ is $K_r$-minor
free: 
even if the constraint graph $G$ is grid-like and planar, the label-extended graph
$H$ can contain a $K_{\Omega(n)}$ minor, even when the alphabet size
is just two. This means that applying the theorems in
Section~\ref{sec:decomp} blindly does not give better than a
$\OO(\log n)$-factor approximation.

This is often a technical issue in analyzing algorithms for Unique Games:
It is most natural to work on the label extended graph
but the label extended graph does not necessarily share the nice properties
in the original graph~\cite{Kolla}.
It is not obvious how to apply low diameter graph
decomposition directly in the original constraint graph to do the
rounding.  For example, in the graph shown in
Figure~\ref{fig:useless}, $x$ is an integral solution but the shortest
path distance (using $x_e$ as the edge weight of $e$) is $0$ for all
pairs of vertices, providing no useful information about which pairs of
vertices we need to separate.

Our main observation is that the shortest path distances are not
useful only when there are edges with large
$x_e$.  In Lemma~\ref{lem:struc}, we prove that if $x_e < 1/2$ for
every $e$, then every odd cycle contains a pair of vertices $u,v$ with
shortest path distance greater
than $1/4$ (using $x_e$ as the edge weight of $e$).
 Therefore, if we apply low diameter graph decomposition with
 $\Delta = 1/4$,
then we can ensure that no odd cycle will remain in any cluster,
and the above calculation shows that the
total weight of the deleted edges is $\OO(r) \cdot \LP^\star$.
To reduce to the case where there are no edges with $x_e \geq 1/2$, we
can simply delete all such edges as their total weight is at most $2
\LP^\star$.  This preprocessing step is remotely similar to some
iterative rounding algorithms~(see~\cite{IterativeRounding}).
See Figure~\ref{fig:useful} for an illustration.

\subsection{Rounding Algorithm} \label{sec:rounding1}

\begin{algorithm}{{\sf Min-Uncut}}{A feasible solution $x$ to \ref{lp:min-uncut} with value $\LP^\star$ on a $K_r$-minor free graph.}{An integral solution to \ref{lp:min-uncut} with total cost $O(r) \cdot \LP^\star$}
  \begin{enumerate}
  \item Let $F_1$ be the subset of edges with $x_e \geq 1/2$.
        Delete all edges in $F_1$ from the graph.
  \item Set the weight $w_e$ of each edge $e$ in the remaining graph
    to be $x_e$.\\ Sample a weakly $(1/4)$-bounded $\OO(r)$-separating
    partition $P$ guaranteed by Theorem~\ref{thm:copsrobbers} in the
    remaining graph.
  \item Let $F_2$ be the set of inter-cluster edges in $P$, i.e. edges $uv$ with $P(u) \neq P(v)$.\\
        Return $F_1 \cup F_2$ as the output.
  \end{enumerate}
\end{algorithm}

\subsection{Main Lemma}

The following lemma allows us to apply low diameter graph decomposition in the original constraint graph.
The proof uses the simple but crucial fact that if we ``shortcut'' an odd cycle, one of the two cycles created is an odd cycle.

\begin{lemma}\label{lem:struc}
Let $G'$ be a graph with edge weight $x_e$ for each edge $e$.
Suppose every odd cycle $C$ has total weight at least $1$,
i.e. $\sum_{e \in C} x_e \geq 1$.
If $0 \leq x_e < \delta \leq 1$ for every edge $e \in G'$,
then every odd cycle $C$ in $G'$ contains a pair of vertices $u, v$
satisfying $d_x(u, v) > (1 - \delta)/2$,
where $d_x(u,v)$ denotes the shortest path distance from $u$ to $v$ induced by the edge weights $x_e$.
\end{lemma}
\begin{proof}
Let $C$ be an arbitrary odd cycle and let $v_0$ be an arbitrary vertex in $C$.
We will prove the stronger statement that if $d_x(v_0,v) \leq (1-\delta)/2$ for every $v \in C$, then there is an edge $e \in C$ with $x_e \geq \delta$.
Note that the contrapositive of this stronger statement clearly implies the lemma.

Since all odd cycles have total weight at least 1, any nontrivial odd walk (may visit some vertices multiple times) from $v_0$ to $v_0$ has total weight at least 1.
This is because any odd walk can be decomposed into edge-disjoint simple cycles, with at least one of which is odd.

We will prove the statement by an inductive argument.
In a general inductive step $t \geq 0$, we maintain a walk $C^{(t)}$ from $v_0$ to $v_0$ satisfying the following properties (see Figure~\ref{fig:short}):
\begin{enumerate}
\item $C^{(t)}$ is a nontrivial odd walk from $v_0$ to $v_0$,
consisting of three paths $P^{(t)}_1$-$P^{(t)}_C$-$P^{(t)}_2$.
\item $P^{(t)}_1$ and $P^{(t)}_2$ contain $v_0$, with $v_0$ being the first vertex of $P^{(t)}_1$ and $v_0$ being the last vertex of $P^{(t)}_2$.
\item Both $P^{(t)}_1$ and $P^{(t)}_2$ have total weight at most $(1-\delta)/2$.
\item $P^{(t)}_C$ is a continuous segment of $C$, i.e. if $C=(v_0,v_1,\ldots,v_k=v_0)$, then $P^{(t)}_C = (v_i, \ldots, v_j)$ for some $0 \leq i < j \leq k$.
In particular, $P^{(t)}_C \neq \emptyset$.
\end{enumerate}
Initially, $C^{(0)}$ is just the cycle $C$, with $P^{(0)}_1=P^{(0)}_2=\emptyset$ and $P^{(0)}_C=C$.

\begin{figure}
  \centering
  \begin{minipage}[t]{0.4\textwidth}
    \centering
    \begin{tikzpicture}[scale=0.6]
      \tikzstyle{invisible} = [fill=black!20!blue, inner sep=0pt, draw=black!20!blue, fill=black!20!blue]
      \begin{scope}[every node/.style={draw, circle, fill=gray!20, ultra thick}]
        \node (v0) at (0, 0) [label=below:$v_0$] {};
        \node (vLeft) at (-2, 3) {};
        \node (vRight) at (2, 3) {};
      \end{scope}

    \begin{scope}
      \path [->, ultra thick] (v0) edge node[left] {$P_1^{(t)}$} (vLeft);
      \path [<-, ultra thick] (v0) edge node [right] {$P_2^{(t)}$}  (vRight);
      \draw[->, ultra thick, line width=3pt, black!20!blue] (vLeft) to [bend left] (0, 5) to [bend left] (vRight);
    \end{scope}
  \end{tikzpicture}
  \end{minipage}
  \begin{minipage}[t]{0.4\textwidth}
    \centering
    \begin{tikzpicture}[scale=0.6]
      \begin{scope}[every node/.style={draw, circle, fill=gray!20, ultra thick}]
        \node (v0) at (0, 0) [label=below:$v_0$] {};
        \node (v1) at (1, 0) [label=below:$v_0$] {};
        \node (vLeft) at (-2, 3) {};
        \node (vRight) at (3, 3) {};
        \node (u) at (0, 5) [label=above:$u$] {};
        \node (u1) at (1, 5) [label=above:$u$] {};
      \end{scope}
      \begin{scope}
        \path [->, ultra thick] (v0) edge node[left] {$P_1^{(t)}$} (vLeft);
        \path [bend left, ->, line width=3pt, black!20!blue] (vLeft) edge node[left, black] {$P_{C_1}^{(t)}$} (u);
        \path [<-, line width=3pt, black!50] (v0) edge node[left, black] {$Q$} (u);
        \path [->, line width=3pt, black!50] (v1) edge node[right, black] {$Q$} (u1);
        \path [<-, ultra thick] (v1) edge node[right] {$P_2^{(t)}$} (vRight);
        \path [bend left, ->, line width=3pt, black!20!blue] (u1) edge node[right, black] {$P_{C_2}^{(t)}$} (vRight);
      \end{scope}
    \end{tikzpicture}
  \end{minipage}
  \caption{The paths involved in the proof of Lemma
    \ref{lem:struc}. The shortcut $Q$ is highlighted gray, and the
    cycle segments $P_{C_j}^{(t)}$ are highlighted blue. Since the
    walk we maintain is odd, one of the two walks we consider in the induction step (right figure) should be odd.}\label{fig:short}
\end{figure}
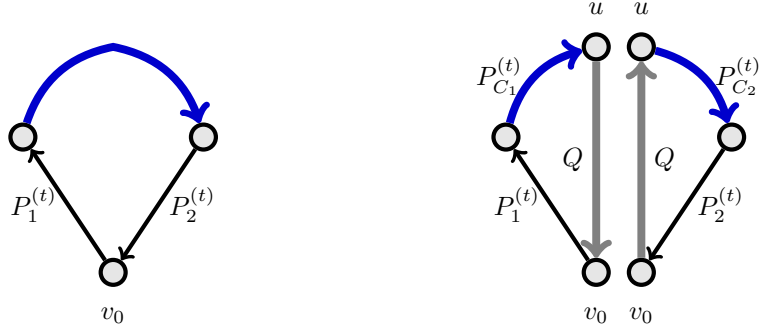

Let $w(P)$ denote the total weight of a path $P$, and let $|P|$ denote the number of edges in $P$.
Since $w(P^{(t)}_1), w(P^{(t)}_2) \leq (1-\delta)/2$, we must have $w(P^{(t)}_C) \geq \delta$, as $C^{(t)}$ is a nontrivial odd walk and thus the total weight is at least one.
The inductive step is to show that if $d_x(v_0,v) \leq (1-\delta)/2$ for all $v \in C$,
then we can construct $C^{(t+1)}$ from $C^{(t)}$ so that $C^{(t+1)}$ still satisfies the properties but $|P^{(t+1)}_C| < |P^{(t)}_C|$.
By applying this inductively, we will eventually construct a walk $C^{(T)}$ that satisfies the properties and $|P^{(T)}_C|=1$, and so $P^{(T)}_C$ is an edge of $C$ with weight $w(P^{(t)}_C) \geq \delta$, and this will complete the proof.

It remains to prove the inductive step (see Figure~\ref{fig:short}).
Let $C^{(t)}$ be a walk that satisfies the properties but $|P^{(t)}_C| \geq 2$.
Let $u$ be an internal vertex of $P^{(t)}_C$, which splits $P^{(t)}_C$ into $P^{(t)}_{C_1}$ and $P^{(t)}_{C_2}$, so that the walk $C^{(t)}$ consists of $P^{(t)}_1$-$P^{(t)}_{C_1}$-$P^{(t)}_{C_2}$-$P^{(t)}_2$.
Since $d_x(v_0,u) \leq (1-\delta)/2$,
there is a path $Q$ from $v_0$ to $u$ with $w(Q) \leq (1-\delta)/2$.
The path $Q$ splits the walk $C^{(t)}$ into two walks, $P^{(t)}_1$-$P^{(t)}_{C_1}$-$Q$ and $Q$-$P^{(t)}_{C_2}$-$P^{(t)}_2$
As $C^{(t)}$ is an odd walk,
a simple parity argument implies that
exactly one of these two walks must be odd,
say $P^{(t)}_1$-$P^{(t)}_{C_1}$-$Q$ (the other case is similar).
Then we let $C^{(t+1)} := P^{(t)}_1$-$P^{(t)}_{C_1}$-$Q$,
with $P^{(t+1)}_1 := P^{(t)}_1$, $P^{(t+1)}_C := P^{(t)}_{C_1}$,
and $P^{(t+1)}_1 := Q$.
It is straightforward to check that $C^{(t+1)}$ still satisfy all the properties and furthermore $|P^{(t+1)}_C| < |P^{(t)}_C|$, completing the proof of the induction step.
\end{proof}

\subsection{Proof of Corollary~\ref{cor:minuncut}} \label{sec:cor}

We are now ready to prove that the algorithm in Section~\ref{sec:rounding1} is an $O(r)$-approximation algorithm for {\sf Min-Uncut}.
In step $1$, since each edge $e$ in $F_1$ has $x_e \geq 1/2$, the total cost of edges in $F_1$ is
\[\sum_{e \in F_1} c_e \leq 2 \sum_{e \in F_1} c_e x_e \leq 2 \LP^\star.\]
Let $G' := G-F_1$ be the remaining graph.
By Lemma~\ref{lem:struc}, every odd cycle of $G'$ contains a pair of vertices $u,v$ with shortest path distance greater than $1/4$.
Let $\Delta=1/4$.
In a $(1/4)$-bounded partition $P$, no cluster can contain an odd cycle $C$
as otherwise the pair of vertices $u,v \in C$ with $d_x(u,v) > 1/4$ guaranteed by Lemma~\ref{lem:struc} would contradict that the cluster has weak diameter at most $1/4$.
So, each cluster induces a bipartite graph, and thus $G' - F_2$ is a bipartite graph where $F_2$ is the set of inter-cluster edges.
Therefore, $F_1 \cup F_2$ is an integral solution to the {\sf Odd Cycle Transversal} problem, and hence an integral solution to the {\sf Min-Uncut} problem.

To complete the proof, it remains to bound the cost of the edges in $F_2$.
We use Theorem~\ref{thm:copsrobbers} to sample from a distribution of partitions which is $\Delta$-bounded and $\OO(r)$-separating, and by definition (\ref{e:cut}) the probability of an edge $e$ being an inter-cluster edge is at most $\OO(r) \cdot x_e / \Delta = \OO(r) \cdot x_e$.
Therefore, the expected cost of $F_2$ is
\[\Exp\sqbr*{\sum_{e \in F_2} c_e}
= \sum_{e=uv \in G'} c_e \cdot \Pp_{P \sim \Par}[P(u) \ne P(v)]
= \sum_{e \in G'} c_e \cdot \OO(r) \cdot x_e = \OO(r) \sum_{e \in G'} c_e x_e
\leq \OO(r) \cdot \LP^\star.\]
Hence, the expected total cost of edges in $F_1 \cup F_2$ is $\OO(r) \cdot \LP^\star$, and this concludes the proof of Corollary~\ref{cor:minuncut} about $K_r$-minor free graphs.
For bounded genus graphs, the proof is the same except that we use Theorem~\ref{thm:cr2} which guarantees the partitioning scheme is $\OO(\log g)$-separating.

\section{$\M2Lin_k$} \label{sec:max2lin}

In this section, we show that the {\sf Min-Uncut} algorithm can be
readily generalized to the $\M2Lin_k$ problem.  The proofs will be
almost identical, so we just highlight the subtle differences.

One important feature of Theorem~\ref{thm:g2linalg} is that the
approximation ratio does not depend on the alphabet size.  The reason
is that the symmetry of the linear constraints allows us to define
inconsistent cycles in the original constraint graph, which will play
the same role as the odd cycles in the {\sf Min-Uncut} problem. This
allows us to reduce $\M2Lin_k$ to the {\sf Inconsistent Cycle
  Transversal} problem.



\subsection{Problem Formulation}

Consider the $\M2Lin_k$ problem where each constraint is of the form
$x_u - x_v = c_{uv} \pmod k$ where $c_{uv} \in \ZZ_k$.  The symmetry
property that we will exploit is that every permutation constraint
$\pi_{uv}$ satisfies: $\pi_{uv}(i + c) = \pi_{uv}(i) + c$ for all $i,c
\in {\mathbb Z}_k$.  Note that there are ``directions'' in the
constraints, as $\pi_{uv} = (\pi_{vu})^{-1}$ and they are in general
different.  In the {\sf Max-Cut} (or {\sf Min-Uncut}) problems, we
have $\pi_{uv} = \pi_{vu}$ as the alphabet set is of size two, and so
the concept of direction was not discussed.


\begin{definition}[Inconsistent cycles for $\M2Lin_k$]
Let $\Ins = (G, \Pi)$ be an instance of $\M2Lin_k$.
A cycle $(v_0, v_1, \ldots, v_l = v_0)$ of length $l$ in $G$ is called inconsistent if
\begin{equation}
  \pi_{v_l v_{l-1}} \circ \pi_{v_{l-1} v_{l-2}} \circ \cdots \circ
  \pi_{v_1 v_0}\neq {\rm Id}\label{eq:incons}
\end{equation} where ${\rm Id}$ is the
identity permutation.
By the aforementioned symmetry property of $\M2Lin_k$, if the
product $\pi$ of permutation constraints along a cycle is not the
identity permutation, then $\pi(i) \neq i$ for all $i \in \ZZ_k$. This
is the crucial property that we will use.

\end{definition}

The following lemma shows that $\M2Lin_k$ is equivalent to the
{\sf Inconsistent Cycle Transversal} problem.  The reason is that
whether a cycle is satisfiable is independent of which label to assign
to the starting vertex because of the symmetry property.  Note that
this does not hold for general Unique Games.
\begin{lemma} \label{fac:inconsistent}
  A $\M2Lin_k$ instance $\Ins = (G, \Pi)$ is satisfiable if and only if $G$ contains no inconsistent cycles.
\end{lemma}
\begin{proof}
Suppose $\Ins$ is satisfiable.  Let $x$ be a satisfying assignment.
Consider an arbitrary cycle $C = (v_0, v_1, \ldots, v_l = v_0)$.  The
permutation constraints on $C$ enforce that $\pi_{v_l v_{l-1}} \circ
\pi_{v_{l-1} v_{l-2}} \circ \cdots \circ \pi_{v_1 v_0} (x(v_0)) =
x(v_0)$ where $x(v_0)$ is the value of $v_0$ in the assignment $x$.
By the symmetry property of the constraints, this implies that
$\pi_{v_l v_{l-1}} \circ \pi_{v_{l-1} v_{l-2}} \circ \cdots \circ
\pi_{v_1 v_0}$ is the identity permutation, and thus it is consistent.

Suppose $G$ has no inconsistent cycles.
Then we show that $G$ is satisfiable by the following trivial
algorithm.  Pick an arbitrary vertex $v_0 \in G$, and set $x(v_0)$ an
arbitrary value.  Then we propagate this assignment to every other
vertex $v$ by using an arbitrary path $P=(v_0, v_1, \ldots, v_l=v)$
from $v_0$ to $v$ and set $x(v) = \pi_{v_l v_{l-1}} \circ \pi_{v_{l-1}
  v_{l-2}} \circ \cdots \circ \pi_{v_1 v_0}(v_0)$.  In particular, we
can use a breadth first search tree to propagate the assignment.
Since $G$ has no inconsistent cycles, any two paths $P_1, P_2$ from
$v_0$ to $v$ will define the same value $x(v)$, as otherwise following
$P_1$ from $v_0$ to $v$ and following $P_2$ from $v$ to $v_0$ will
give us an inconsistent cycle.  This implies that any non-tree
constraint $uv$ is also satisfied by the assignment, as otherwise it
means that there are two paths from $v_0$ to $v$ defining different
values from $x(v)$, one path being the tree path from $v_0$ to $u$
plus the edge $uv$, and the other path being the tree path from $v_0$
to $v$.
\end{proof}

\subsection{Linear Programming Relaxation}

Given Lemma~\ref{fac:inconsistent}, we can formulate the
minimization version of the $\M2Lin_k$ problem, the $\Mi2Lin_k$ problem,
as the {\sf Inconsistent Cycle Transversal} problem, where the
objective is to find a subset of edges of minimum cost that intersects
all the inconsistent cycles.  We can then use the same linear
programming relaxation for the {\sf Min-Uncut} problem, with $\Cc$
being the set of inconsistent cycles in the constraint graph.  Again,
we can design a polynomial time separation oracle to check whether a
solution $x$ is feasible, by constructing the label extended graph and
using shortest path computations as in Section~\ref{sec:LP}
(see~\cite{GTUGC}).

\subsection{Rounding Algorithm and Analysis}

The rounding algorithm is exactly the same as in Section~\ref{sec:rounding1},
and so we do not repeat it here.
The analysis is also the same,
which relies on a generalization of Lemma~\ref{lem:struc}.

\begin{lemma}\label{lem:strucGen}
Let $G'$ be a graph with edge weight $x_e$ for each edge $e$.
Suppose every inconsistent cycle $C$ has total weight at least $1$,
i.e. $\sum_{e \in C} x_e \geq 1$.
If $0 \leq x_e < \delta \leq 1$ for every edge $e \in G'$,
then every inconsistent cycle $C$ in $G'$ contains a pair of vertices $u, v$
satisfying $d_x(u, v) > (1 - \delta)/2$,
where $d_x(u,v)$ denotes the shortest path distance from $u$ to $v$ induced by the edge weights $x_e$.
\end{lemma}
\begin{proof}
  The proof is essentially identical, by replacing every occurrence of
  ``odd'' by ``inconsistent''.  The only place that needs explanation
  is in the last paragraph of Lemma~\ref{lem:struc}, when we split an
  inconsistent walk using a path $Q$ from $v_0$ to $u$ into two walks
  $P^{(t)}_1$-$P^{(t)}_{C_1}$-$Q$ and $Q$-$P^{(t)}_{C_2}$-$P^{(t)}_2$,
  and we need to argue that at least one of these two walks is
  inconsistent.  Suppose both walks are consistent.  Let $\pi_{P_1}$
  be the composition of the permutation constraints from $v_0$ to $u$
  following the path $P^{(t)}_1$-$P^{(t)}_{C_1}$, $\pi_Q$ be the
  composition of the permutation constraints from $u$ to $v_0$
  following the path $Q$, and $\pi_{P_2}$ be the composition of the
  permutation constraints from $u$ to $v_0$ following the path
  $P^{(t)}_{C_2}$-$P^{(t)}_2$.  The first walk is consistent means
  that $\pi_Q \circ \pi_{P_1} = {\rm Id}$, and the second walk is
  consistent means that $\pi_{P_2} \circ (\pi_{Q})^{-1} = {\rm Id}$.
  But this implies that following the first walk and then the second
  walk is consistent, and thus the original walk is also consistent as
  ${\rm Id} = (\pi_{P_2} \circ (\pi_{Q})^{-1}) \circ (\pi_Q \circ
  \pi_{P_1}) = \pi_{P_2} \circ \pi_{P_1}$,
  contradicting that the original walk is
  inconsistent.  The rest of the proof is identical.
\end{proof}

With Lemma~\ref{lem:strucGen}, using exactly the same argument as in
Section~\ref{sec:cor} gives us the proof of
Theorem~\ref{thm:g2linalg}.

\section{General Unique Games}\label{sec:genug}

For general Unique Games, we could not reduce the problem to some cycle
cutting problem in the original constraint graph.  Instead, we modify
the LP-based algorithm of Gupta and Talwar~\cite{GTUGC} to prove Theorem~\ref{thm:ugkalg}.

\subsection{Linear Programming Relaxation}

Gupta and Talwar~\cite{GTUGC} use the following linear programming relaxation for the Unique Games problem.

\noindent
\begin{boxedminipage}{\textwidth}
  {\begin{flalign*}
    \min \LP^\star = \sum_{uv \in E} \frac{c_{uv}}{2} \sum_{l = 1}^k
    d(u, v, l) \tag{LP-UG}\label{lp:ug} &&
  \end{flalign*}
  subject to
  \begin{flalign*}{}
    \sum_{l = 1}^k x(u, l) & = 1 & \forall u \in V &\\
    d(u, v, l) & \ge \Abs{x(u, l) - x(v, \pi_{uv}(l))} & \quad \forall uv \in E,~l \in [k] & \\
    \sum_{i = 1}^t d(v_{i - 1}, v_i, l_{i - 1}) & \ge x(u, l_0) & \quad \forall C,~\forall u \in C,~\forall l_0 \in B_{u, C}\label{eq:cons} & \\
    1 \ge x(u, l) & \ge 0 & \forall u \in V,~\forall l \in [k]
  \end{flalign*}}%
\end{boxedminipage}
\newline

The intended value of $x(u,l)$ is $1$ if we assign the label $l$ to
vertex $u$ and $0$ otherwise, and so the first constraint enforces
that we assign exactly one label to each vertex.  The intended value
of $d(u,v,l)$ is $1$ if we assign $u$ to $l$ but not assign $v$ to
$\pi_{uv}(l)$ or vice versa and is $0$ otherwise.
So $\sum_{l=1}^k d(u,v,l)$ is two if the constraint $\pi_{uv}$ is not
satisfied and is $0$ if the constraint is satisfied, and therefore
the objective function is to minimize the total cost of the violated
constraints.  The third constraint is the inconsistent cycle
constraint in the label extended graph: $B_{u,C}$ is defined as the
set of ``bad'' labels at $u$, so that if $u$ is assigned some label in
$B_{u,C}$, then propagating this label along the cycle must violate
some permutation constraint in $C$.  So, the intention of the third
constraint is that if we assign some label in $B_{u,C}$ to vertex $u$,
then the number of violated constraint along the cycle $C$ must be at
least $1$.  This is similar to our inconsistent cycle constraint, but
defined on the label extended graph.

\subsection{Proof Overview} \label{sec:overview2}

Gupta and Talwar~\cite{GTUGC} gave a polynomial time randomized algorithm to
return an integral solution of cost $\OO(\log n) \cdot \LP^\star$ from
a feasible solution to the LP with objective value $\LP^\star$.

The main technique in their rounding algorithm is the use of a
low average distortion tree to propagate an assignment
from a vertex.  Their propagation rounding algorithm picks an
arbitrary vertex $u \in V$ and assigns it a random label $l_u$
according to the probability distribution defined by $x(u,l)$.  Then
they design a correlated sampling scheme to sample a label $l_v$ for a
neighbor $v$ of $u$ satisfying the properties that
$\Pp[l_v = l] = x(v,l)$ and
$\Pp[l_v \neq \pi_{uv}(l_u)] \leq \sum_{l=1}^k d(u,v,l)$.
 They use this
correlated sampling to propagate the assignment from the starting
vertex to every vertex in the graph using the low average distortion tree.
Their approximation ratio comes from the average distortion $\OO(\log n)$ of
the tree given by the FRT embedding~\cite{FRT}, which can not be
improved even for planar graphs.

We will still use the propagation rounding method of Gupta and Talwar,
but we apply it to different trees.  In~\cite{GTUGC}, the tree $T$
needs not be a spanning tree in the constraint graph (i.e. some edges
in the tree may not exist in the graph), and this adds some
complication to the analysis.  In our application, all tree edges will
be graph edges and we can use a simpler lemma in their proof.  For an
edge $uv \in E$, we let $d_G(u,v) := \sum_{l=1}^k d(u,v,l)$, and let
$d_T(u,v) := \sum_{xy \in P} d_G(x,y)$ where $P$ is the unique path
from $u$ to $v$ in the tree $T$.

\begin{lemma}[Lemma 3.1 in \cite{GTUGC}]\label{lemma:main}
  Let $x$ be the assignment produced by the propagation rounding algorithm using correlated sampling along a tree $T$.
For every edge $uv \in G$, we have
  \begin{equation*}
    \Pp[x(v) \ne \pi_{uv}(x(u))] \le d_G(u, v) + 2 d_T(u, v).
  \end{equation*}%
\end{lemma}%

The idea of our algorithm is very simple.
We use the strongly $\Delta$-bounded $\OO(r^2)$-separating partitioning scheme to decompose the graph, using $d_G(u,v)$ as the weight of edge $uv \in E(G)$.
As each cluster is of strong diameter $\Delta$,
we simply use a shortest path tree in each cluster to do the propagation rounding and apply Lemma~\ref{lemma:main} to prove Theorem~\ref{thm:ugkalg}.
We will choose $\Delta$ to balance the losses in the two steps.


\subsection{Rounding Algorithm}

\begin{algorithm}{$\UG_k$}{A feasible solution $x,d$ to LP-UG with value $\LP^\star$ on a $K_r$-minor free graph.}{An integral solution to LP-UG with total cost $O(r) \cdot \sqrt{\LP^\star}$.}
  \begin{enumerate}
  \item Set the weight $w_{uv}$ of each edge $uv$ to be $d_G(u,v)$.\\
        Sample a strongly $\Delta$-bounded $\OO(r^2)$-separating partition $P$ guaranteed by Theorem~\ref{thm:crstrong}.
  \item Let $F$ be the set of inter-cluster edges in $P$, i.e. edges $uv$ with $P(u) \neq P(v)$.\\
        Delete $F$ from $G$.
  \item In each cluster $C_j$ in the remaining graph, compute a shortest path tree $T_j$.
  \item Run Gupta-Talwar propagation rounding on each cluster $C_j$ using tree $T_j$.
  \item Return the solution $x,d$ as the union of the solution in each cluster.
  \end{enumerate}
\end{algorithm}

\subsection{Proof of Theorem~\ref{thm:ugkalg}}

Since the partitioning scheme is $\OO(r^2)$-separating,
by definition \eqref{e:cut},
each edge $e$ is deleted with probability
\[
\Pp[{\rm edge~} uv \textrm{ is deleted}] = \OO(r^2) \cdot \frac{d_G(u, v)}{\Delta}.
\]
Hence, the expected total cost of the deleted edges in Step $2$ is
\[\sum_{uv \in E} c_{uv} \cdot \Pp[{\rm edge~} uv {\rm~is~deleted}]
= \OO(r^2/\Delta) \sum_{uv \in E} c_{uv} \cdot d_G(u,v) =
\OO(r^2/\Delta) \cdot \LP^\star.\] We just assume that all of these
edges will be violated by the assignment we produce at the end.  Since
each cluster $C_j$ has strong diameter $\Delta$, the shortest path
tree $T_j$ satisfies
\[d_{T_j}(u, v) \le \Delta \quad \forall u, v \in C_j.\]
Using the Gupta-Talwar propagation rounding, by
Lemma~\ref{lemma:main}, each edge in cluster $C_j$ is violated with
probability $O(\Delta)$, and therefore the total cost of the violating
constraints in the Step $4$ is at most $O(\Delta) \sum_{e \in E}
c_e$.  By choosing $\Delta = r \cdot \sqrt{ \LP^\star / \sum_{e \in E}
  c_e}$, the total cost of the violating constraints is at most $r \cdot
\sqrt{\LP^\star \cdot \sum_{e \in E} c_e}$.  When $\LP^\star = \ee
\cdot \sum_{e \in E} c_e$, the total cost of the violating constraint
is at most $r \sqrt{\ee} \sum_{e \in E} c_e$, proving
Theorem~\ref{thm:ugkalg} for $K_r$-minor free graphs.  For bounded
genus graphs, we just use the bound in Theorem~\ref{thm:cr2} to
replace $r^2$ by $\log g$, and the same proof gives
Theorem~\ref{thm:ugkalg} for bounded genus graphs.

\section{Discussions and Open Problems}

The algorithm for general Unique Games has a similar structure to the
subexponential time algorithm~\cite{ABS}.  Both algorithms first
deletes a small fraction of edges so that each remaining component has
some nice properties, and then solve the problem in each component
using a propagation rounding method.  The nice property in~\cite{ABS}
is that each component has few small eigen\-values (which qualitatively
means that the components have good expansion property), and the
decomposition result is based on random walks.  The nice property in
this paper is that each component has small diameter, and the
decomposition result is based on some combinatorial methods.
The key to these algorithms is
some graph decomposition result.  Is there some property that captures
both good expansion and small diameter so that graph decomposition is
still possible?
Is there some property that captures both good expansion and small diameter so
that propagation rounding still works?

Another open question is whether the ideas in this paper can be
generalized to handle graphs with many small eigenvalues.

\subparagraph*{Acknowledgements.}

We thank Tsz Chiu Kwok, Akshay Ramachandran and Hong Zhou for useful
discussions.


\begin{thebibliography}{10}

\bibitem{CopsRobbers}
Ittai Abraham, Cyril Gavoille, Anupam Gupta, Ofer Neiman, and Kunal Talwar.
\newblock Cops, robbers, and threatening skeletons: Padded decomposition for
  minor-free graphs.
\newblock In {\em Proceedings of the 46th Annual ACM Symposium on Theory of
  Computing}. ACM, 2014.

\bibitem{UncutSDP}
Amit Agarwal, Moses Charikar, Konstantin Makarychev, and Yury Makarychev.
\newblock O(sqrt(log n)) approximation algorithms for min uncut, min 2cnf
  deletion, and directed cut problems.
\newblock In {\em Proceedings of the 37th Annual ACM Symposium on Theory of
  Computing}. ACM, 2005.

\bibitem{AgarwalUGC}
Naman Agarwal.
\newblock Unique games conjecture: The boolean hypercube and connections to
  graph lifts.
\newblock Master's thesis, University of Illinois at Urbana-Champaign, 2014.

\bibitem{hypercube}
Naman Agarwal, Guy Kindler, Alexandra Kolla, and Luca Trevisan.
\newblock Unique games on the hypercube.
\newblock {\em Chicago Journal of Theoretical Computer Science}, 2015, 2015.

\bibitem{ABS}
Sanjeev Arora, Boaz Barak, and David Steurer.
\newblock Subexponential algorithms for unique games and related problems.
\newblock In {\em Proceedings of the 51st Annual IEEE Symposium on Foundations
  of Computer Science}. IEEE, 2010.

\bibitem{ExpUGC}
Sanjeev Arora, Subhash~A Khot, Alexandra Kolla, David Steurer, Madhur Tulsiani,
  and Nisheeth~K Vishnoi.
\newblock Unique games on expanding constraint graphs are easy.
\newblock In {\em Proceedings of the 40th Annual ACM Symposium on Theory of
  Computing}. ACM, 2008.

\bibitem{Baker}
Brenda~S Baker.
\newblock Approximation algorithms for np-complete problems on planar graphs.
\newblock {\em Journal of the ACM}, 41(1), 1994.

\bibitem{BFK}
Nikhil Bansal, Uriel Feige, Robert Krauthgamer, Konstantin Makarychev,
  Viswanath Nagarajan, Joseph Naor, and Roy Schwartz.
\newblock Min-max graph partitioning and small set expansion.
\newblock In {\em Proceedings of the 52nd Annual IEEE Symposium on Foundations
  of Computer Science}. IEEE, 2011.

\bibitem{BRS}
Boaz Barak, Prasad Raghavendra, and David Steurer.
\newblock Rounding semidefinite programming hierarchies via global correlation.
\newblock In {\em Proceedings of the 52nd Annual IEEE Symposium on Foundations
  of Computer Science}. IEEE, 2011.

\bibitem{Bartal}
Yair Bartal.
\newblock Probablistic approximation of metric spaces and its algorithmic applications.
\newblock In {\em Proceedings of the 37-th Annual IEEE Symposium on Foundations of Computer Science}. IEEE, 1996.

\bibitem{Deform1}
Punyashloka Biswal, James~R Lee, and Satish Rao.
\newblock Eigenvalue bounds, spectral partitioning, and metrical deformations
  via flows.
\newblock {\em Journal of the ACM}, 57(3), 2010.

\bibitem{CMM0}
Moses Charikar, Konstantin Makarychev, and Yury Makarychev.
\newblock Near-optimal algorithms for unique games.
\newblock In {\em Proceedings of the 38th Annual ACM Symposium on Theory of
  Computing}. ACM, 2006.

\bibitem{CMM}
Eden Chlamtac, Konstantin Makarychev, and Yury Makarychev.
\newblock How to play unique games using embeddings.
\newblock In {\em Proceedings of the 47th Annual IEEE Symposium on Foundations
  of Computer Science}. IEEE, 2006.

\bibitem{DHK}
Erik~D Demaine, Mohammad~Taghi Hajiaghayi, and Ken-ichi Kawarabayashi.
\newblock Algorithmic graph minor theory: Decomposition, approximation, and
  coloring.
\newblock In {\em Proceedings of the 46th Annual IEEE Symposium on Foundations
  of Computer Science}. IEEE, 2005.

\bibitem{BiDim}
Erik~D Demaine and MohammadTaghi Hajiaghayi.
\newblock The bidimensionality theory and its algorithmic applications.
\newblock {\em The Computer Journal}, 51(3), 2008.

\bibitem{FRT}
Jittat Fakcharoenphol, Satish Rao, and Kunal Talwar.
\newblock A tight bound on approximating arbitrary metrics by tree metrics.
\newblock In {\em Proceedings of the 35th Annual ACM Symposium on Theory of
  Computing}. ACM, 2003.

\bibitem{FaTa}
Jittat Fakcharoenphol and Kunal Talwar.
\newblock An improved decomposition theorem for graphs excluding a fixed minor.
\newblock In {\em Approximation, Randomization, and Combinatorial
  Optimization.. Algorithms and Techniques}. Springer Berlin Heidelberg, 2003.

\bibitem{MinUncutPlanarLP}
Jean Fonlupt, Ali~Ridha Mahjoub, and JP~Uhry.
\newblock Compositions in the bipartite subgraph polytope.
\newblock {\em Discrete mathematics}, 105(1-3), 1992.

\bibitem{GTUGC}
Anupam Gupta and Kunal Talwar.
\newblock Approximating unique games.
\newblock In {\em Proceedings of the 70th Annual ACM-SIAM Symposium on Discrete
  Algorithms}. SIAM, 2006.

\bibitem{GS}
Venkatesan Guruswami and Ali~Kemal Sinop.
\newblock Lasserre hierarchy, higher eigenvalues, and approximation schemes for
  graph partitioning and quadratic integer programming with psd objectives.
\newblock In {\em Proceedings of the 52nd Annual IEEE Symposium on Foundations
  of Computer Science}. IEEE, 2011.

\bibitem{Hadlock}
Frank Hadlock.
\newblock Finding a maximum cut of a planar graph in polynomial time.
\newblock {\em SIAM Journal on Computing}, 4(3), 1975.

\bibitem{DeformBG}
Jonathan~A Kelner.
\newblock Spectral partitioning, eigenvalue bounds, and circle packings for
  graphs of bounded genus.
\newblock {\em SIAM Journal on Computing}, 35(4), 2006.

\bibitem{Deform2}
Jonathan~A Kelner, James~R Lee, Gregory~N Price, and Shang-Hua Teng.
\newblock Metric uniformization and spectral bounds for graphs.
\newblock {\em Geometric and Functional Analysis}, 21(5), 2011.

\bibitem{KhotUG}
Subhash Khot.
\newblock On the power of unique 2-prover 1-round games.
\newblock In {\em Proceedings of the 34th Annual ACM Symposium on Theory of
  computing}. ACM, 2002.

\bibitem{G2LinHard}
Subhash Khot, Guy Kindler, Elchanan Mossel, and Ryan O'Donnell.
\newblock Optimal inapproximability results for max-cut and other 2-variable
  csps?
\newblock {\em SIAM Journal on Computing}, 37(1), 2007.

\bibitem{KhReVC}
Subhash Khot and Oded Regev.
\newblock Vertex cover might be hard to approximate to within 2- $\varepsilon$.
\newblock {\em Journal of Computer and System Sciences}, 74(3), 2008.

\bibitem{KPRDecomp}
Philip Klein, Serge~A Plotkin, and Satish Rao.
\newblock Excluded minors, network decomposition, and multicommodity flow.
\newblock In {\em Proceedings of the 25th Annual ACM Symposium on Theory of
  Computing}. ACM, 1993.

\bibitem{Kolla}
Alexandra Kolla.
\newblock Spectral algorithms for unique games.
\newblock {\em Computational Complexity}, 20(2), 2011.

\bibitem{IterativeRounding}
Lap~Chi Lau, Ramamoorthi Ravi, and Mohit Singh.
\newblock {\em Iterative methods in combinatorial optimization}, volume~46.
\newblock Cambridge University Press, 2011.

\bibitem{LeeSid}
James~R Lee and Anastasios Sidiropoulos.
\newblock Genus and the geometry of the cut graph
\newblock In {\em Proceedings of the 21st Annual ACM-SIAM Symposium on Discrete
  Algorithms}. SIAM, 2010.

\bibitem{ellipsoidref}
L{\'a}szlo~Lov{\'a}sz, Martin~Gr{\"o}tschel, and Alexander Schrijver.
\newblock Geometric algorithms and combinatorial optimization.
\newblock {\em Berlin: Springer-Verlag}, 33, 1988.

\bibitem{RaghUGC}
Prasad Raghavendra.
\newblock Optimal algorithms and inapproximability results for every csp?
\newblock In {\em Proceedings of the 40th Annual ACM Symposium on Theory of
  Computing}. ACM, 2008.

\bibitem{RSSSE}
Prasad Raghavendra and David Steurer.
\newblock Graph expansion and the unique games conjecture.
\newblock In {\em Proceedings of the 42nd ACM Symposium on Theory of
  Computing}. ACM, 2010.

\bibitem{RSTSSE}
Prasad Raghavendra, David Steurer, and Madhur Tulsiani.
\newblock Reductions between expansion problems.
\newblock In {\em Proceedings of the 27th IEEE Annual Conference on
  Computational Complexity}. IEEE, 2012.

\bibitem{SpectralWorks}
Daniel~A Spielman and Shang-Hua Teng.
\newblock Spectral partitioning works: Planar graphs and finite element meshes.
\newblock {\em Linear Algebra and its Applications}, 421(2-3), 2007.

\bibitem{SteuVish}
David Steurer and Nisheeth~K Vishnoi.
\newblock Connections between unique games and multicut.
\newblock In {\em Electronic Colloquium on Computational Complexity},
  volume~16, 2009.

\bibitem{TrevUGC}
Luca Trevisan.
\newblock Approximation algorithms for unique games.
\newblock In {\em Proceedings of the 46th Annual IEEE Symposium on Foundations
  of Computer Science}. IEEE, 2005.


\end{thebibliography}
\end{document}